\documentclass{article}

\usepackage[nocompress]{cite}
\usepackage{graphicx}
\usepackage[utf8x]{inputenc}
\usepackage{amsmath}

\usepackage{amssymb,amsthm,mathcomp,stmaryrd}
\usepackage[caption=false,font=normalsize,labelfont=sf,textfont=sf]{subfig}
\usepackage{tikz,pgfplots}
\usepackage{pdfsync}
\usepackage{cancel}
\usepackage{url}
\usepackage{bookmark}
\usepgfplotslibrary{groupplots}
\usepgfplotslibrary{external}
\usetikzlibrary{plotmarks}

\DeclareMathOperator*{\argmax}{\arg\!\max}
\DeclareMathOperator*{\diag}{diag}

\newtheorem{theorem}{Theorem}
\newtheorem{lemma}{Lemma}
\newtheorem{definition}{Definition}
\newtheorem{proposition}{Proposition}
\newtheorem{remark}{Remark}
\newtheorem{corollary}{Corollary}

\title{Heterogeneous Differential Privacy}

\author{Mohammad Alaggan, S\'ebastien Gambs, Anne-Marie Kermarrec}

\begin{document}
\maketitle

\begin{abstract}
The massive collection of personal data by personalization systems has rendered the preservation of privacy of individuals more and more difficult. Most of the proposed approaches to preserve privacy in personalization systems usually address this issue uniformly across users, thus ignoring the fact that users have different privacy attitudes and expectations (even among their own personal data). In this paper, we propose to account for this non-uniformity of privacy expectations by introducing the concept of heterogeneous differential privacy. This notion captures both the variation of privacy expectations among users as well as across different pieces of information related to the same user. We also describe an explicit mechanism achieving heterogeneous differential privacy, which is a modification of the Laplacian mechanism by Dwork, McSherry, Nissim, and Smith. In a nutshell, this mechanism achieves heterogeneous differential privacy by manipulating the sensitivity of the function using a linear transformation on the input domain.
Finally, we evaluate on real datasets the impact of the proposed mechanism with respect to a semantic clustering task. The results of our experiments demonstrate that heterogeneous differential privacy can account for different privacy attitudes while sustaining a good level of utility as measured by the recall for the semantic clustering task.
\end{abstract}
\section{Introduction}  
The amount of personal information about individuals exposed on the Internet is increasing by the second. While such data may be used for recommendation and personalization purposes \cite{DBLP:conf/middleware/BertierFGKL10,Zeng:2012:UID:2189736.2189749,ZhouEtAl2012,Wen:2010:AOI:1772690.1772875,1264820}, this also raises serious privacy concerns. At first, leveraging personal information to enhance the user experience through personalization services might seems contradictory with the preservation of the privacy of users of such systems. However in recent years, several approaches have been proposed to rely on Privacy-Enhancing Technologies (PETs), whose aim is to preserve privacy while maintaining a good level of utility for the proposed personalization service \cite{springerlink:10.1007/s11257-011-9110-z,DBLP:conf/opodis/AlagganGK11,AGK12,McSherry:2009:DPR:1557019.1557090,springerlink:10.1007/978-0-387-70992-5_4}. One popular approach whose objective is to provide strong privacy guarantees despite auxiliary information that the adversary could have is the concept of differential privacy \cite{DBLP:conf/tamc/Dwork08,DBLP:conf/crypto/MironovPRV09,limits-two-party-dp-2011-MMPRTV,McSherry:2009:DPR:1557019.1557090,DworkNaor2010,DBLP:conf/tcc/DworkMNS06,DBLP:journals/corr/abs-1103-2626}. 
 
Most of these approaches implicitly assume homogeneity by considering that users have uniform privacy requirements. However, in an environment composed of a myriad of communities, such as the Internet, it is highly plausible that users have heterogeneous privacy attitudes and expectations. For instance, consider a collaborative social platform in which each user is associated to a profile (\emph{e.g.}, a set of URLs that a user has tagged in a system such as Delicious\footnote{\url{http://del.icio.us/}}). It is natural to expect that for a particular user some items in his profile are considered more sensitive by him than others, thus calling for a system that can deal with different privacy requirements across items. Similarly, Alice might be more conservative about her privacy than Bob, requiring different privacy requirements across users. 

This non-uniformity of privacy attitudes has been acknowledged by major social networking sites~\cite{Preibusch09privacy-preservingfriendship,IMC/Gummadi11}. For instance in Facebook, a user can now set individual privacy settings for each item in his profile. However in this particular example, privacy is mainly addressed by restricting, through an access-control mechanism, who is allowed to access and view a particular piece of information. Our approach can be considered to be orthogonal but complementary to access-control. More precisely, we consider a personalized service, such as recommendation algorithm, and we enforce the privacy requirements of the user on its output. Heterogeneous privacy requirements might also arise with respect to pictures, depending on the location in which the picture was taken or the persons appearing on it\cite{IMC/Gummadi11}. In the future, users are likely to expect item-grained privacy for other services\footnote{Note that systems supporting item-grained privacy can also provide user-grained privacy (\emph{i.e.}, for instance by setting the privacy level of all items in some user's profile to the same value in the privacy setting of this user), and therefore the former can be considered as a generalization of the latter. However, this assumes that the privacy weights have a global meaning across the entire system, and are not defined only relative to a user.}.

Furthermore, as highlighted by Zwick and Dholakia in 1999 \cite{Zwick99modelsof} and as evidenced by anthropological research, privacy attitudes are highly dependent on social and cultural norms.  A similar point was raised in 2007 by Zhang and Zhao in a paper on privacy-preserving data mining \cite{ZhangZhao2007} in which they mentioned that in practice it is unrealistic to assume homogeneous privacy requirements across a whole population. In particular, their thesis is that enforcing the same privacy level across all users and for all types of personal data could lead to an unnecessary degradation of the performance of such systems as measured in terms of accuracy.  More specifically, enforcing the same privacy requirements upon all users (even those who do not require it) might degrade the performance in comparison to a system in which strict privacy requirements are only taken into account for those who ask for it. The same type of argument can also be made for different items of the same user. Hence, designing a system supporting heterogeneous privacy requirements could lead to a global improvement of the performance of this system as compared to a homogeneous version. Therefore, the main challenge is to be able to account for the variety of privacy requirements when leveraging personal data for recommendation and personalization.

In this paper, we address this challenge through the introduction of the concept of \emph{heterogeneous differential privacy}, which considers that the privacy requirements are not homogeneous across users and items from the same user (thus providing item-grained privacy). This notion can be seen as an extension of the concept of differential privacy \cite{DBLP:conf/tamc/Dwork08} introduced originally by Dwork in the context of databases. We also describe an explicit mechanism achieving heterogeneous differential privacy, which we coin as the ``stretching mechanism''. We derive a bound on the distortion introduced by our mechanism, which corresponds to a distance between the expected output of the mechanism and the original value of the function to be computed.  Finally, we conduct an experimental evaluation of our mechanism on a semantic clustering task using real datasets. The results obtained show that the proposed approach can still sustain a high utility level (as measured in terms of recall) while guaranteeing heterogeneous differential privacy.

The outline of the paper is as follows. First, in Section~\ref{sec:background}, we describe the background of differential privacy as well as some preliminaries on matrices and sets necessary to understand our work. Afterwards in Section~\ref{sec:heter-diff-priv}, we introduce the novel concept of heterogeneous differential privacy along with the description of an explicit mechanism achieving it. Then, we assess experimentally the impact of the proposed mechanism by evaluating it on a semantic clustering task in Section~\ref{sec:hdp-practice}. In Section~\ref{sec:related-work}, we present the related work on heterogeneous privacy mechanisms before concluding with a discussion on the actual limitations of the approach as well as possible extensions in Section~\ref{sec:conclusion}.

\section{Background}
\label{sec:background}

In this section, we briefly introduce the background on differential privacy (Section~\ref{diff_privacy}) as well as some notions on matrices and sets that are necessary to understand the concept of heterogeneous differential privacy (Section~\ref{prelim}).

\subsection{Differential Privacy}
\label{diff_privacy}

We begin with providing some background of differential privacy, which was originally introduced by Dwork  \cite{DBLP:conf/tamc/Dwork08} in the context of statistical databases.
The main guarantee provided by this approach is that if a differentially private mechanism is applied on a database composed of the personal data of individuals, no output would become significantly more (or less) probable whether or not a participant removes this particular data from the dataset. In a nutshell, it means that for an adversary observing the output of the mechanism, the advantage gained from the presence (or absence) of a particular individual in the database is negligible. This statement is a statistical property about the behavior of the mechanism (\emph{i.e.}, function) and holds independently of the auxiliary knowledge that the adversary might have gathered. More specifically, even if the adversary knows the whole database but one individual row, a mechanism satisfying differential privacy still protects the privacy of this row. The parameter $\varepsilon$ is public and may take different values depending on the application (for instance it could be $0.01$, $0.1$, $0.25$ or even $2$). While it is sometimes difficult to grasp the intuition about the significance of a particular value for $\varepsilon$ \cite{DBLP:conf/isw/LeeC11}, a smaller value of $\varepsilon$ implies a higher privacy level.

Differential privacy was originally designed for ensuring privacy to individuals who have contributed with their personal data to the construction of a statistical database. In this setting, each individual is a row (\emph{i.e.}, coordinate) in this database (\emph{i.e.}, vector). Differential privacy guarantees that \emph{almost} no difference will be observed to the output of query performed on the database, whether or not the individual (a single row) has contributed to the database by submitting his data, and therefore this information is considered as being protected. 

When the database is the profile of a user, which is a vector of items (sometimes called \emph{the micro-data setting}), the whole vector (\emph{i.e.}, database) is owned by a single individual. This difference impacts the interpretation that can be done when speaking about \emph{protecting the privacy of this individual}. In particular, contrary to the first setting of statistical database, an individual does not have the choice to submit or not his data. Rather, if he chooses not to use his profile as input to the collaborative social system, then he will not benefit from the service. However in this new setting, the user is still left with the possibility of selecting a subset of items in his profile before participating. In this case, the main objective of differential privacy is to ensure that when a user adds or removes a single item from his profile, this has a small effect on the output of the computation. However, one caveat is that if the profile of the user contains nothing but items related to a particular sensitive topic (\emph{e.g.}, cancer), then in order to get at least a little bit of utility that information has to be leaked. This observation is in line with the impossibility result of Dwork and Naor stating that if a privacy-preserving mechanism provides any utility, then it has to cause a privacy breach whose magnitude is at least proportional to the min-entropy of the utility \cite{DworkNaor2010}. Thus, this limitation is true for any possible privacy-preserving mechanism and is not inherent to the micro-data setting (\emph{i.e.}, this limitation also holds for the database setting).

The difference of a single row between two profiles can be defined formally through the concept of \emph{neighboring profiles}.
Each user is associated with a profile representing his personal data, which can be defined as a vector in $\mathbb{R}^n$ (for some $n$ fixed for all users across the system). This representation is generic enough to encompass a variety of possible user profiles. For instance, restricting the domain to $\{0,1\}^n$ can be used to represent a binary string (which is a universal representation) or a subset of items of a global domain of items.

\begin{definition}[Neighboring profile]
\label{def:neigh-prof}
Two profiles $\vec d , \vec d^{(i)} \in \mathbb{R}^n$ are said to be neighbors if there exists an item $i \in \{1,\dots,n\}$ such that $d_k=d^{(i)}_k$ for all items $k\neq i$. This neighborly relation is denoted by $\vec d \sim \vec d^{(i)}$. 
\end{definition}
An equivalent definition states that $\vec d$ and $\vec d^{(i)}$ are neighbors if they are identical except for the $i$-th coordinate. For instance, the profiles $(0,1,2)$ and $(0,2,2)$ are neighbored while the profiles $(0,1,2)$ and $(0,2,3)$ are not. Differential privacy can be defined formally in the following manner.

\begin{definition}[$\epsilon$-differential privacy \cite{DBLP:conf/tamc/Dwork08}]
A randomized function $\mathcal M: \mathbb{R}^n \to \mathbb{R}$ is said to be $\varepsilon$-differentially private if for all neighboring profiles $\vec d \sim \vec d^{(i)} \in \mathbb{R}^n$, and for all outputs $t \in \mathbb{R}$ of this randomized function, the following statement holds: 
\begin{equation}
\Pr[\mathcal M(\vec d) = t] \leqslant \exp(\varepsilon) \Pr[\mathcal M(\vec d^{(i)}) = t] \enspace ,
\end{equation}
 in which $\exp$ refers to the exponential function.
\end{definition}

Differential privacy aims at reducing the contribution that any single coordinate of the profile can have on the output of a function. The maximal magnitude of such contribution is captured by the notion of (global) \emph{sensitivity}.

\begin{definition}[Global sensitivity \cite{DBLP:conf/tcc/DworkMNS06}]
The global sensitivity $S(f)$ of a function $f$ is the maximum absolute difference obtained on the output over all neighboring profiles:  
\begin{equation}
S(f) = \max\limits_{\vec d \sim \vec d^{(i)}} | f(\vec d) - f(\vec d^{(i)}) | \enspace ,
\end{equation}
 where $\vec d \sim \vec d^{(i)}$ means that $\vec d$ and $\vec d^{(i)}$ are neighboring profiles (\emph{cf.} Definition~\ref{def:neigh-prof}). 
\end{definition}

Dwork proposed a technique called the \emph{Laplacian mechanism} \cite{DBLP:conf/tcc/DworkMNS06} that achieves $\varepsilon$-differential privacy by adding noise to the output of a function proportional to its global sensitivity.  The noise is distributed according to the Laplace distribution (with PDF $\frac{1}{2\sigma}\exp(-|x|/\sigma)$, in which $\sigma=S(f)/\varepsilon$ is a scale parameter).

The novel mechanism that we propose in this paper (to be detailed later) achieves heterogeneous differential privacy by modifying the sensitivity of the function to be released (and therefore the function itself) before applying the standard Laplacian mechanism.

\subsection{Preliminaries}
\label{prelim}

Before delving into the details of our approach, we need to briefly introduce some preliminary notions on matrices and sets such as the concept of \emph{shrinkage matrix} \cite{matOp10}. A shrinkage matrix is a linear transformation that maps a vector to another vector with less magnitude, possibly distorting it by changing its direction.

\begin{definition}[Shrinkage matrix]
A matrix $A$ is called a shrinkage matrix if and only if $A=\diag(\alpha_1,\dots,\alpha_n)$ such that each diagonal coefficient is in the range $0 \leqslant \alpha_i \leqslant 1$.
\end{definition}

For example, the matrix 
\begin{equation}
\left( \begin{array}{ccc}
0.7 & 0 & 0 \\
0 & 0.3 & 0 \\
0 & 0 & 1 \end{array} \right)
\end{equation}
is a shrinkage matrix.

\begin{definition}[Semi-balanced set]
A set $D\subseteq\mathbb{R}^n$ of column vectors is semi-balanced if and only if for all shrinkage matrices $A=\diag(\alpha_1,\dots,\alpha_n)$, and for all $\vec x \in D$, we have $A \vec x \in D$.
\end{definition}
For instance, the set 
\begin{equation}
\{\vec x =(x_1,x_2) \in \mathbb R^2 \mid 0 < x_1, x_2 < 1 \}
\end{equation}
 is a semi-balanced set that can be visualized as a square from $(0,0)$ to $(1,1)$ in the Euclidean plane.

\section{Heterogeneous Differential Privacy}
\label{sec:heter-diff-priv}

In this section, we introduce the novel concept of \emph{heterogeneous differential privacy} (HDP). We start by giving the necessary definitions in Section~\ref{sec:hdp-definitions}, before describing in Section~\ref{sec:stretching-mechanism} how to construct the Stretching Mechanism, which ensures heterogeneous differential privacy. More precisely, we first detail how to construct the privacy-preserving estimator in Section~\ref{sec:constr-estim}. Afterwards, we discuss why and how the privacy vector expressing the privacy expectations of a user should also be kept private in Section~\ref{sec:privacy-psi-privacy}. Finally, an upper bound on the distortion induced by the Stretching Mechanism is provided in Section~\ref{sec:distortion}. 

\subsection{Definitions}
\label{sec:hdp-definitions}

We now define HDP-specific notions such as the concept of \emph{privacy vector}, which is a key notion in HDP. This vector contains the privacy requirements of each coordinate (\emph{i.e.}, item) in the input profile (\emph{i.e.}, vector) of a user, and is defined as follows.

\begin{definition}[Privacy vector] \label{def:v} Given a profile $\vec d \in D$ in which $D$ is a semi-balanced set of column vectors composed of $n$ coordinates, let ${\vec v} \in [0,1]^n$ be the privacy vector associated with the profile $\vec d$. The owner of item $d_i$ is responsible for choosing the privacy weight $v_i$ associated to this item (by default $v_i$ is set to be $1$ if it was not explicitly specified by the owner). A privacy weight $v_i$ of zero corresponds to \emph{absolute privacy} while a value of $1$ refers to \emph{standard privacy}, which in our setting directly correspond to the classical definition of $\varepsilon$-differential privacy.
\end{definition}

The mere presence of the privacy vector introduces potential privacy breaches, thus this vector should also be protected. Thus, we need to ensure that in addition to the profile, the privacy vector $\vec v$ also remains private, such that each entry $v_i$ of this vector should only be known by its owner. Otherwise, the knowledge of a privacy weight of a particular item might leak information about the profile itself. For instance, learning that some items have a high privacy weight may reveal that the user has high privacy expectations for and is therefore interested in this specific type of data. We define \emph{heterogeneous differential privacy} in the following manner.

\begin{definition}[$(\varepsilon,\vec v)$-differential privacy]
 A randomized function $\mathcal M: D \to \mathbb{R}$ is said to be $(\varepsilon,\vec v)$-differentially private if for all items $i$, for all neighboring profiles $\vec d\sim\vec d^{(i)}$, and for all possible outputs $t \in \mathbb{R}$ of this function, the following statement holds:
\begin{equation}
\Pr[ \mathcal M (\vec d) = t ] \leqslant \exp(\varepsilon v_i) \Pr [ \mathcal M (\vec d^{(i)}) = t ], \enspace
\end{equation}
in which $\exp$ refers to the exponential function.
\end{definition}

Since a privacy weight $v_i \leqslant 1$, heterogeneous differential privacy implies the standard notion of $\varepsilon$-differential privacy as shown by the following remark. 

\begin{remark}[Equivalence of $(\varepsilon,\vec v)$-DP and $\varepsilon$-DP.]
  Let $\overline\varepsilon = \varepsilon \overline v$ and $\underline\varepsilon = \varepsilon \underline v$, such that $\overline v=\max_i v_i$ (the maximum privacy weight) and $\underline v=\min_i v_i$ (the minimum privacy weight). Then, we have:
$ \underline\varepsilon\mbox{-DP} \implies (\varepsilon,\vec v)\mbox{-DP} $ and $ (\varepsilon,\vec v)\mbox{-DP} \implies \overline\varepsilon\mbox{-DP}$. As a consequence, $(\varepsilon,\vec 1)\mbox{-DP}$ holds \emph{if and only if} $\varepsilon$-DP also holds, in which $\vec 1=(1,\cdots,1)$.
\end{remark}

Finally, we rely on a variant of the notion of global sensitivity, implicitly introduced \cite[Lemma~1]{DBLP:journals/corr/abs-1111-2885}, that we call \emph{modular global sensitivity}. 

\begin{definition}[Modular global sensitivity \cite{DBLP:journals/corr/abs-1111-2885}]
The modular global sensitivity $S_i(f)$ is the global sensitivity of $f$ when $\vec d$ and $\vec d^{(i)}$ are neighboring profiles that differ on exactly the item $i$.
\end{definition}

In a nutshell, the modular global sensitivity reflects the maximum difference that a \emph{particular item} $i$ can cause by varying its value (over its entire domain) while keeping all other items fixed.

\subsection{The Stretching Mechanism}
\label{sec:stretching-mechanism}

Thereafter, we describe a generic mechanism achieving heterogeneous differential privacy that we coin as the \emph{Stretching Mechanism}. We assume that the privacy preferences for each item are captured through a privacy vector $\vec v$ (\emph{cf.} Definition~\ref{def:v}). Given an arbitrary \emph{total} function $f: D \to \mathbb{R}$, in which $D$ is a semi-balanced set of columns vectors of $n$ coordinates, and whose global sensitivity $S(f)$ is finite, we construct a randomized function $\hat{f}(\vec d,\vec v,\varepsilon)$ \emph{estimating} $f$ while satisfying $(\varepsilon,\vec v)$-differential privacy.

Before delving into the details of this method, we provide a little intuition on how and why it works. 
A lemma in \cite[Lemma~1]{DBLP:journals/corr/abs-1111-2885} asserts
that the Laplacian mechanism $\mathcal{M}(\vec d)=f(\vec
d)+\mathsf{Lap}(\sigma)$ with mean $0$ and standard deviation $\sigma$
provides 
\begin{equation}
  \Pr[\mathcal{M}(\vec d) = t] \leqslant \exp(\varepsilon_i)
\Pr[\mathcal{M}(\vec d^{(i)}) = t] \enspace ,
\end{equation}
 where $\varepsilon_i=S_i(f)/\sigma$. In other words, differential privacy can be achieved by setting the perturbation induced by the Laplacian mechanism to be proportional to the modular global sensitivity \cite{DBLP:journals/corr/abs-1111-2885} instead of the standard global sensitivity. Therefore, a natural approach for enforcing heterogeneous differential privacy is to manipulate the modular global sensitivity $S_i(f)$ by modifying the function $f$ itself.

\subsubsection{Constructing the Estimator}
\label{sec:constr-estim}

Let $T : [0,1]^n \to \mathbb R^{n\times n}$ be a function taking as input a privacy vector $\vec v$ and returning as output a shrinkage matrix, with the property that $T(\vec 1)=I$, such that $I$ is the identity matrix and $\vec 1=(1,\cdots,1)$.  Let  also $R$ be a mapping sending a function $f: D \to \mathbb R$ and a privacy vector $\vec v\in \mathbb [0,1]^n$ to the function $R(f,\vec v):D\to\mathbb R$ with $R(f,\vec v)(\vec d)=f(T(\vec v) \cdot \vec d)$. Recall that the Laplace distribution centered at 0 with scale parameter $\sigma$ has the following probability density function
\begin{equation}
h(x) = \frac{1}{2\sigma}\exp(-|x|/\sigma) \enspace .\label{eq:h}
\end{equation}
Finally, let $N$ be a Laplacian random variable with parameter $\sigma=\sigma(f,\varepsilon)=S(f)/\varepsilon$, in which $S(f)$ refers to the global sensitivity of the function $f$
and $\varepsilon$ the privacy parameter. The following statement proves that this \emph{Stretching Mechanism} $R$ satisfies heterogeneous differential privacy.

\begin{theorem}[Achieving HDP via stretching mecanism]\label{thm:hdp}
Given a privacy vector $\vec v$, if the function $T(\vec v)$ satisfies $S_i(R(f,\vec v)) \leqslant v_i S(f)$ then the randomized function $\hat{f}(\vec d,\vec v,\varepsilon) = R(f,\vec v)(\vec d) + N$ satisfies $(\varepsilon,\vec v)$-differential privacy.
\begin{proof}
See Appendix.
\end{proof}
\end{theorem}

In a nutshell, $T(\vec v)$ is a shrinkage matrix, whose shrinking factor in each coordinate is computed independently of all other coordinates. More precisely, the shrinking factor for a particular item depends only on the privacy weight associated to this coordinate. The value used by the mechanism is the lowest amount of shrinkage (\emph{i.e.}, distortion) still achieving the target modular global sensitivity of that coordinate. In the following section we provide an explicit construction of $T(\vec v)$ for which we prove that by Lemma~\ref{lem:gsi_leq_psii_gs} the condition of Theorem~\ref{thm:hdp} is satisfied, and therefore that $\hat f$ achieves $(\varepsilon,\vec v)$-differential privacy. 

\subsubsection{Computing the Shrinkage Matrix}
\label{sec:constructing-tv}

The HDP mechanism $\hat f(\vec d,\vec v, \varepsilon)$ adds Laplacian noise to a modified function $R(f,\vec v)(\vec d)=f(T(\vec v)\cdot \vec d)$. In this section, we specify how to construct $T(\vec v)$ such that $\hat f$ satisfies HDP. Thereafter, we use $R$ to denote $R(f,\vec v)$ for the sake of simplicity.  Let $T(\vec v)=\diag(\vec w)$ for some $\vec w \in [0,1]^n$ to be computed from the privacy vector $\vec v$ and $S(R,\vec w)$ be the sensitivity of \mbox{$R=f(T(\vec v)\cdot \vec d)=f(\diag(\vec w)\cdot \vec d)$} given $\vec w$. Similarly, let $S_i(R,\vec w)$ be the modular global sensitivity of $R$ given $\vec w$. We denote by $(\vec w_{-i},w_i^\prime)$ the vector resulting from replacing the item $w_i$ in $\vec w$ to $w_i^\prime$ (\emph{e.g.}, $(\vec 1_{-i},w_i)=(1,\dots,w_i,\dots,1)$). Each $w_i$ can be computed from $v_i$ by solving the following optimization problem:
  \begin{equation}
  \begin{array}{lc}
    \max & w_i \enspace , \\
    \mbox{subject to: } & S_i(R,(\vec 1_{-i},w_i)) \leqslant v_i S(f) \enspace.
  \end{array}\label{eq:opt}  
\end{equation}

Note that a solution satisfying this constraint always exists and is reached by setting $w_i$ to $0$. The $w_i$'s are never released after they have been computed locally by the rightful owner, and the modular global sensitivity $S_i(R)$ is only used in the proof and is not revealed to the participants, in the same manner as the noise generated. The participants only have the knowledge of the global sensitivity $S(f)$.
Thus, the only way in which the profile $\vec d$ could leak is through its side effects to the output, which we prove to achieve $\varepsilon$-DP in Theorem~\ref{thm:priv-priv-sett}.

\begin{lemma}
  \label{lem:gsi_leq_psii_gs}
  If $T(\vec v)=\diag(\vec w)$ such that for all $i$: 
  \begin{equation}
 S_i(R,(\vec 1_{-i},w_i)) \leqslant v_i S(f)
\end{equation}
 (the constraint of \eqref{eq:opt}) then $R$ satisfies: 
  \begin{equation}
 S_i(R, \vec w)\leqslant v_i S(f)
\end{equation} 
for all $i$.
\begin{proof}
See Appendix.
\end{proof}
\end{lemma}

\subsubsection{Hiding the Privacy Vector}
\label{sec:privacy-psi-privacy}
 
By themselves, the privacy weights could lead to a privacy breach if there are release publicly \cite{DBLP:conf/sigecom/GhoshR11,DBLP:journals/corr/abs-1111-2885}. For instance, learning that the user has set a high weight on a particular item might be indicated that the user possesses this item on his profile and that he has a high privacy expectation about it. Thus, the impact of the privacy weights on the observable output of the mechanism should be characterized. 
The following theorem states that when the profile $\vec d$ is fixed, the randomized function $\hat{f}$ satisfies $\varepsilon$-differential privacy over neighboring privacy vectors $\vec v \sim\vec v^{(i)}$. Thus, the privacy vector can also be considered to be hidden and protected by the guarantees of differential privacy.
   
\begin{theorem}[Protecting the privacy vector with $\varepsilon$-DP]
\label{thm:priv-priv-sett}
The randomized function $\hat{f}$ provides $\varepsilon$-differential privacy for each individual privacy weight of $\vec v$. This means that for all neighboring privacy vectors $\vec v \sim \vec v^{(i)}$, for all outputs $t \in \mathbb R$ and profiles $\vec d$, the following statement holds:
\begin{equation}
 \Pr[\hat{f} (\vec d,\vec v,\varepsilon) = t] \leqslant \exp(\varepsilon) \Pr[\hat{f}(\vec d,\vec v^{(i)},\varepsilon) = t] \enspace .
\end{equation}
 \begin{proof}
 See Appendix.
 \end{proof}
\end{theorem}

\subsubsection{Estimating the Distortion Induced by HDP}
\label{sec:distortion}

Intuitively, if $\underline w$ is the minimum of $\vec w$ (the diagonal of $T(\vec v)$), and $\vec d$ is the input profile, then the distortion introduced by \emph{stretching} the function as measured in terms of absolute additive error is bounded by $1-\underline w$ times the norm of $\vec d$ multiplied by the norm of the gradient of the \emph{semi-stretched} function at $\vec d$.

More formally, let $f$ be a continuous and differentiable function on a semi-balanced set $D$, and let $(\vec v, \vec d) \in [0,1]^n\times D$ be respectively, the privacy vector and the profile considered. The following theorem provides a bound on the distortion introduced on the output by modifying the global sensitivity of the function $f$ as done by the HDP (\emph{i.e.}, stretching) mechanism described previously.

\begin{theorem}[Bound on the distortion induced by the Stretching Mechanism]
\label{thm:distortion}
Let $f:D\to \mathbb R$ be a function from a semi-balanced set $D$ to the reals, and let $\vec v\in [0,1]^n$ be a privacy vector and $T:[0,1]^n \to \mathbb{R}^{n\times n}$ be a function taking a privacy vector to a shrinkage matrix. Finally, let $R$ be a mapping sending a function $f$ and a privacy vector $\vec v$ to the function $R(f,\vec v): D\to \mathbb R$ such that 
$R(f,\vec v)(\vec d)=f(T(\vec v) \cdot \vec d)$ for all vectors $\vec d$. The distortion (\emph{i.e.}, distance) between $f$ and $R(f,\vec v)$ is bounded by:
  \begin{equation}
  \lvert f(\vec d)-R(f,\vec v)(\vec d) \rvert \leqslant
  \max\limits_{0\leqslant c \leqslant 1} (1-\underline w) \lVert \nabla f(B
\cdot \vec d) \rVert \lVert \vec d \rVert \enspace ,
\end{equation}
where $B=c I + (1-c) T(\vec v)$, $\underline w = \min_i w_i$ is the minimum of $\vec w$ (the diagonal of $T(\vec v)$), and $\nabla f$ is the gradient of the function $f$.
\begin{proof}
See Appendix.
\end{proof}  
\end{theorem}

This bound is particularly useful in situations in which the norm of the gradient of the function $f$ is bounded from above by a constant.
However, even if the norm of the gradient is not bounded by a constant, the bound can still be useful. For instance, in the case of the scalar product function, the bound on the distortion will be $(1-\underline w)\lVert \vec d\rVert^2$, due to the fact that the gradient of the scalar product function is equal to $\lVert B \cdot \vec d \rVert \leqslant \lVert \vec d \rVert$ (since $B$ is a shrinkage matrix).
One restriction on the application of this bound is that the function $f$ to be protected should have a finite global sensitivity, and therefore the scalar product function mentioned has to restrict its domain to be finite, thus preventing the distortion bound from being infinite.

\section{HDP in Practice}
\label{sec:hdp-practice}

To assess the practicality of our approach, we have applied the HDP mechanism on a collaborative social system \cite{DBLP:conf/middleware/BertierFGKL10}, and evaluated its impact on a related semantic clustering task. 
In this collaborative social system, each user (\emph{i.e.} node) is associated with a profile. A profile is the set of items the user has liked or tagged (\emph{e.g.}, the set of URLs in his Delicious account). The objective of the semantic clustering task is to assign each node with the $k$-closest neighbors according to a given similarity metric. In this paper, we use the classical \emph{cosine similarity} (introduced later) to quantify the similarity between two profiles. The task is carried out using a fully distributed protocol, therefore the nodes compute locally (\emph{i.e.}, without relying on a central authority) their similarity with other profiles.

\subsection{Applying HDP to Semantic Clustering}

In the context of distributed semantic clustering, we are interested in providing heterogeneous differential privacy guarantees to the profiles of nodes (\emph{i.e.}, users). More precisely, we consider the scenario in which a particular user can assign a privacy weight, between $0$ and $1$, to each item of his profile. The value $0$ corresponds to the strongest privacy guarantee in the sense that the presence (or absence) of this item will not affect the outcome (the clustering) at all, while the value $1$ is the \emph{lowest} level of privacy possible in our framework (however it still provides the standard guarantees of $\varepsilon$-differential privacy). Thus, the privacy weights of a user directly reflect his privacy attitudes with respect to particular items of his profile, and as a side effect determines the influence of this item in the clustering process. In particular, an item with a higher weight will contribute more to the clustering process, while a item with a lower weight will influence less the resulting clustering.

The \emph{cosine similarity} between two profiles $X$ and $Y$ is defined as 
\begin{equation}
\mathsf{cos\_sim}(X,Y) = \frac{ \lvert X \cap Y \rvert }{\sqrt{ \lvert X \rvert \times \lvert Y \rvert }} \enspace , 
\end{equation}
such that $\lvert X \cap Y \rvert$ is the number of items in common between $X$ and $Y$, and $\lvert X \rvert$ and $\lvert Y \rvert$ correspond to the number of items of $X$ and $Y$, respectively.

The \emph{indicator function} of a profile, when it is represented as a binary vector, for the item $i$ is $1$ if the $i^{th}$ item is present in the profile and $0$ otherwise. More formally, the $i^{th}$ coordinate $\chi_i(x)$ of the indicator function $\chi(x)$ of the profile $x$ is denoted by:
\begin{equation}
\chi_i(x) = \left\{ 
    \begin{array}{ll}
      1    & \quad\text{if $i\in x$} \\
      0    & \quad\text{otherwise}
    \end{array}
  \right.
\enspace .
\end{equation}
Using the notation of the indicator function, the cosine similarity could be defined as
\begin{equation}
 \frac{ \chi(X)\cdot \chi(Y) }{\lVert \chi(X) \rVert_2 \lVert \chi(Y)
   \rVert_2} \enspace , 
\end{equation}
in which the operation ``$\cdot$'' denotes the scalar product. In the following, we apply HDP to the scalar product function and use this modified version to compute the 
cosine similarity on profiles represented as binary vectors.

Given two profiles $X$ and $Y$ and their corresponding indicator functions $\vec x =\chi(X)$ and $\vec y=\chi(Y)$, let $\mathsf{SP}(\vec x,\vec y)=\sum_i x_i y_i$ refers to the scalar product between the two profiles. The privacy vector $\vec v$ is composed of two parts, one for the profile $\vec x$ and the other for the profile $\vec y$: $(\vec v^{\vec x},\vec v^{\vec y})$. Consider the matrix $T(\vec v)=\diag(v)$ and let $R(\mathsf{SP},\vec v)=\mathsf{SP}(T(\vec v^{\vec x})\cdot \vec x,T(\vec v^{\vec y})\cdot \vec y)$ be the Stretching Mechanism, in which $T$ is the stretch specifier. This mechanism $R$ satisfies the premise of Theorem~\ref{thm:hdp} and therefore the choice of $T(\vec v)=\diag(\vec v)$ also ensures HDP, as proven in the following lemma.

\newcommand{\vx}{\vec v^{\vec x}}
\newcommand{\vy}{\vec v^{\vec y}}

\begin{lemma}
\label{lem:diag_v_premis_thm_hdp}
Consider a matrix $T(\vec v)=\diag(\vec v)$ and a mechanism $R(\mathsf{SP},\vec v)=\mathsf{SP}(T(\vec v^{\vec x})\cdot \vec x,T(\vec v^{\vec y})\cdot \vec y)$, such that $\vec x$ and $\vec y$ correspond to profiles and $v^{\vec x}$ and $v^{\vec y}$ to their associated privacy vectors. In this situation, the following statement is always true: $S_i(R(\mathsf{SP},\vec v)) \leqslant v_i S(\mathsf{SP})$ for all $i$.
 \begin{proof}
 See Appendix.
 \end{proof}
\end{lemma}

The previous lemma proves that the proposed modified version of scalar product is differentially private, while the next lemma simply states that if we rely on this differentially private version of scalar product to compute the cosine similarity (or any similar metric), the outcome of this computation will still be differentially private. A standard (\emph{i.e.}, non-heterogeneous) version of the following post-processing lemma can be found in the literature \cite{Kasiviswanathan2008}, which we have generalized to heterogeneous differential privacy.

\begin{lemma}[Effect of post-processing on HDP]
\label{lem:post-processing}
If a randomized function $\hat f$ satisfies $(\varepsilon,\vec v)$-differential privacy, then for any randomized function $g:\mathsf{Range}(\hat f)\to \mathbb R$ independent of the input, the composed function $g\circ \hat f$ satisfies also $(\varepsilon,\vec v)$-differential privacy. The randomness of the function $g$ is assumed to be independent of the randomness of $\hat f$ in order for this property to hold.
\begin{proof}
See Appendix.
\end{proof}
\end{lemma}

\subsection{Experimental Evaluation}

For the experiments, we assume that in real life, nodes will assign different privacy weights to the items in their profiles. In order to simulate this, we generated privacy weights uniformly at random from a set of $n$ equally-spaced values in a fixed range $[\underline u,\overline u]$. More formally, each item is associated with a privacy weight sampled uniformly at random from the set $\{\underline u, \overline u + \delta, \dots, \overline u - \delta, \overline u \}$, $\delta = (\overline u-\underline u)/(n-1)$, for $0\leqslant \underline u < \overline u \leqslant 1$. For instance, if $\underline u=0.5$, $\overline u=1$ and $n=3$, then the weights assigned to items will be uniformly chosen from the set $\{0.5,0.75,1\}$.

We run our experiments on three datasets coming respectively from Delicious, Digg and a survey conducted within our lab. About 120 users participated in the survey and submitted their feedback (in forms of like/dislike) on approximately 200 news. Therefore, in the survey dataset a user's profile consists of the news he has liked, while for the Digg dataset a profile consists of the items that a user has forwarded to others users. Finally, in the Delicious dataset, the profile of the user consists of the items he has tagged.  
\begin{itemize}
\item \emph{Delicious dataset.} Delicious (\url{delicious.com}) is a collaborative platform for keeping bookmarks in which users can tag the URLs of websites they liked. The Delicious dataset consists in the profiles of approximately $500$ users, a profile being a set of URLs that the user has tagged. The total number of URLs in the collective set of users' profiles is over 50,000 URLs. In such a setting, the problem of similarity computation arises naturally, when providing personalized services such as the recommendation of URLs drawn from the ones tagged in Delicious. For the sake of simplicity, in the experiments conducted, each URL was assigned a unique identifier in the range of $\{1,\dots,50000\}$, in order to handle identifiers as integers instead of URL strings. The average size of a profile is $135$ URLs.
\item \emph{Digg dataset.} The dataset consists of $500$ users of Digg (\url{digg.com}), a social news website. The profile of these users is composed of the news that they have shared over a period of $3$ weeks in 2010. All the users considered have shared more than $7$ items per week and the dataset contains $1250$ items, each of which has been shared by at least $10$ users. The average size of a profile is $317$ items.
\item \emph{Survey dataset.} Around $200$ randomly chosen news on various topics have been shown to $120$ colleagues and relatives, who have then submitted their opinion in terms of like/dislike for each news. The average size of the profile is 68. Indeed, while each user has answered to all the 200 pieces of news, he has only liked $68$ of those pieces of news on average.
\end{itemize}

The distributed clustering algorithm is gossip-based and works in an iterative manner \cite{DBLP:conf/middleware/BertierFGKL10}.  In order to assess the quality (\emph{i.e.}, utility) of a particular clustering, we rely on the \emph{recall} metric. The recall can be defined as the ratio between the number of search items a node could find in the profiles of his $k$ closest neighbors (as induced by the clustering) over all possible items of his profile. We consider this metric for our experiments but other standard metrics used in recommendation systems could work as well. In the experiments conducted, the profile of each user is split at random into a training set composed of $90$\% of the profile while the remaining $10$\% is used for testing. After $20$ rounds of exchanging gossip messages during the clustering protocol, each user searches for those $10$\% of items in the profiles of the $k$ closest neighbors provided by the clustering protocol. The recall is then equal to the ratio of items found in the profiles of the neighbors over all the possible items contained in the testing set. The average recall of all users is then reported as the outcome of the experiment.

In Figure~\ref{fig:het_priv_one_group}, we have plotted the three cases for which the interval $(\underline u, \overline u)$ is set to be $(0,1), (0.5,1),$ and $(0.9,1)$. The $x$-axis represents $\underline u$, while the $y$-axis is the recall averaged over all slices (from $n=1$ to $n=10$) for the experiment in the range $[x,1]$. Afterwards in Figure~\ref{fig:het_priv_steps}, we have fixed the range $\underline u\in\{0,0.5,0.9\}$ and $\overline u=1$ and plot the average recall over all users over all runs versus $n$, the number of slices (ranging from $1$ to $10$). In both figures, the error bars represent the variance.

\begin{figure}
  \centering
  \begin{tikzpicture} 
    \begin{groupplot}[ %
      group style={ group name=my plots0, group size=3 by 1, x descriptions at=edge bottom, horizontal sep=2cm, vertical sep=0.1cm, },%
      footnotesize, width=4.5cm, height=4cm, xlabel=$\underline{u}$, ylabel=Recall, xmin=0,xmax=1, xtick={0,0.5,.9},xmajorgrids=true]%
      \nextgroupplot 
        \addplot[mark=o,error bars/.cd,y dir=both, y explicit] plot
        coordinates { (0, 0.962105) +- (0,{0.00675465}) (0.5, 0.957416) +- (0,{0.00970656}) (0.9, 0.968561) +- (0,{0.00829513}) };%
        \addplot[solid] plot coordinates { (0,0.994324) (0.5,0.994324) (0.9,0.994324) };
        \addplot[mark=square] plot coordinates { (0,0.965779) (0.5,0.965779) (0.9,0.965779) };
      \nextgroupplot 
        \addplot[mark=o,error bars/.cd,y dir=both, y explicit] plot
        coordinates { (0, 0.924312) +- (0,{0.00353454}) (0.5, 0.928518) +- (0,{0.00368793}) (0.9, 0.950578) +- (0,{0.00270331}) }; %
        \addplot[solid] plot coordinates { (0, 0.984533) (0.5, 0.984533) (0.9, 0.984533) };
        \addplot[mark=square] plot coordinates { (0, 0.924899) (0.5, 0.924899) (0.9, 0.924899) }; 
      \nextgroupplot[legend entries={HDP;,Baseline;,Totally random},legend to name=groups_legend0] 
        \addplot[mark=o,error bars/.cd,y dir=both, y explicit] plot
        coordinates { (0, 0.122623) +- (0, {0.00735435}) (0.5, 0.211224) +- (0, {0.00947697}) (0.9, 0.224832) +- (0, {0.00835777}) };%
        \addplot[solid] plot coordinates { (0,0.261) (0.5,0.261) (0.9,0.261) }; 
        \addplot[mark=square]         plot coordinates { (0,0.08)  (0.5,0.08)  (0.9,0.08) }; 
    \end{groupplot} 
    \draw (my plots0 c2r1.south)+(0pt,-50pt) node {\ref{groups_legend0}};
    \draw (my plots0 c1r1.north)+(0pt,+10pt) node {Survey};
    \draw (my plots0 c2r1.north)+(0pt,+10pt) node {Digg};
    \draw (my plots0 c3r1.north)+(0pt,+10pt) node {Delicious};
  \end{tikzpicture}
  
  \caption{The value reported is the average recall obtained when all peers have the same distribution over privacy weights for all items, averaged over the number of slices. \emph{Baseline} refers to the recall obtained when the system run with no privacy guarantees using the plain version of the clustering algorithm, while \emph{Random} refers to a random clustering process in which peers choose their neighbors totally at random.}
  \label{fig:het_priv_one_group}
\end{figure}
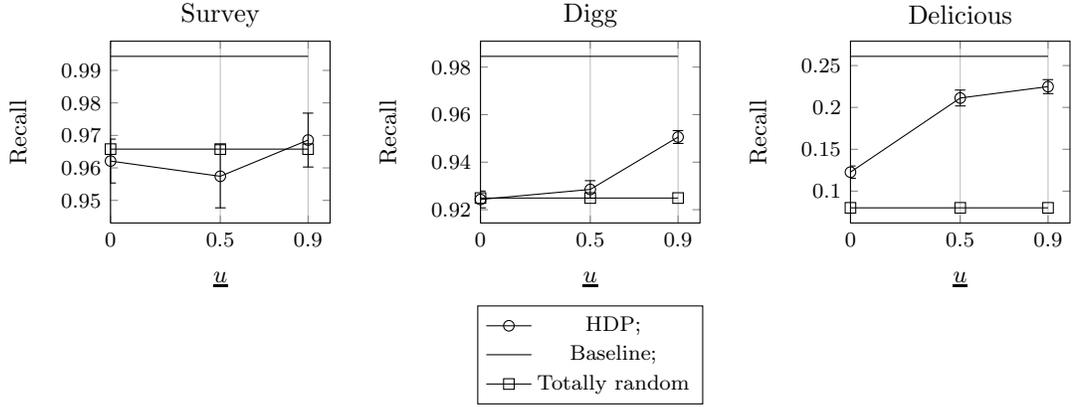

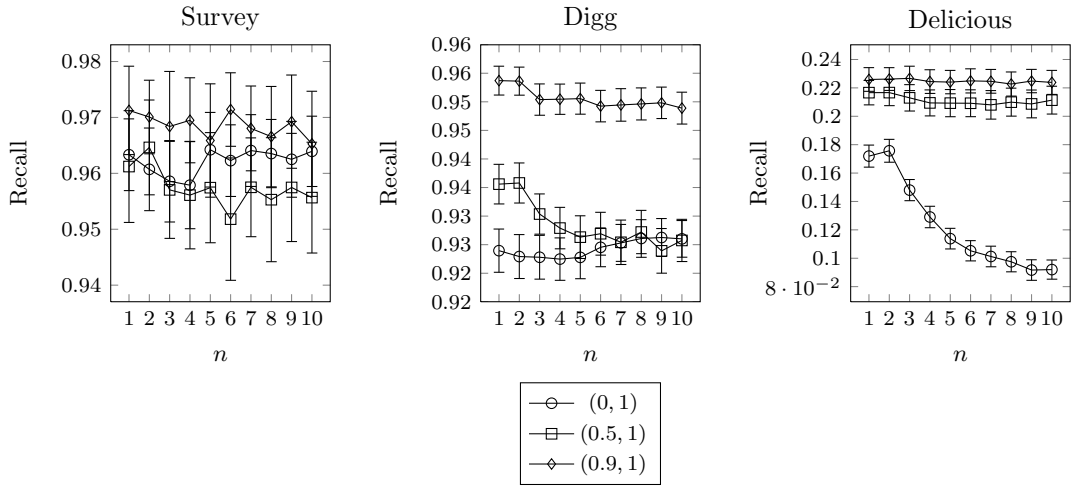
\begin{figure}
  \centering
  \begin{tikzpicture} 
    \begin{groupplot}[ %
      group style={ group name=my plots0, group size=3 by 1, x descriptions at=edge bottom, horizontal sep=2cm, vertical sep=0.1cm, },%
      footnotesize, width=4.5cm, height=5cm, xlabel=$n$, ylabel=Recall, xtick={1,2,3,4,5,6,7,8,9,10}]%
      \nextgroupplot 
        \addplot[mark=o,error bars/.cd,y dir=both, y explicit] plot coordinates {(1,0.963327) +- (0,{0.00641545}) (2,0.960714) +- (0,{0.00739063}) (3,0.958579) +- (0,{0.00728208}) (4,0.957882) +- (0,{0.00778697}) (5,0.964223) +- (0,{0.00667375}) (6,0.962271) +- (0,{0.00640652}) (7,0.964075) +- (0,{0.00641996}) (8,0.963536) +- (0,{0.00610116}) (9,0.962509) +- (0,{0.00678185}) (10,0.963934) +- (0,{0.00628815})};%
        \addplot[mark=square,error bars/.cd,y dir=both, y explicit] plot coordinates {(1,0.961213) +- (0,{0.00999496}) (2,0.964636) +- (0,{0.0084781}) (3,0.957042) +- (0,{0.00867417}) (4,0.956106) +- (0,{0.00961074}) (5,0.957434) +- (0,{0.00985024}) (6,0.951809) +- (0,{0.0109436}) (7,0.957504) +- (0,{0.00883809}) (8,0.955286) +- (0,{0.0110912}) (9,0.957474) +- (0,{0.0096757}) (10,0.955656) +- (0,{0.00990883})};%
        \addplot[mark=diamond,error bars/.cd,y dir=both, y explicit] plot coordinates {(1,0.971259) +- (0,{0.00790364}) (2,0.970075) +- (0,{0.00658133}) (3,0.968379) +- (0,{0.00984721}) (4,0.969489) +- (0,{0.00760315}) (5,0.965872) +- (0,{0.0101376}) (6,0.971411) +- (0,{0.00657448}) (7,0.968037) +- (0,{0.00758275}) (8,0.966562) +- (0,{0.00897953}) (9,0.969245) +- (0,{0.00833922}) (10,0.96528) +- (0,{0.00940239})};%
      \nextgroupplot 
        \addplot[mark=o,error bars/.cd,y dir=both, y explicit] plot coordinates {(1,0.923952) +- (0,{0.00378265}) (2,0.922922) +- (0,{0.00384395}) (3,0.922799) +- (0,{0.00383054}) (4,0.922475) +- (0,{0.00371176}) (5,0.922777) +- (0,{0.00373456}) (6,0.924535) +- (0,{0.00339307}) (7,0.92531) +- (0,{0.00323192}) (8,0.926083) +- (0,{0.00329191}) (9,0.926254) +- (0,{0.0033173}) (10,0.926016) +- (0,{0.00320774})};%
        \addplot[mark=square,error bars/.cd,y dir=both, y explicit] plot coordinates {(1,0.935594) +- (0,{0.00346598}) (2,0.935795) +- (0,{0.00350649}) (3,0.93036) +- (0,{0.00352085}) (4,0.927904) +- (0,{0.0036266}) (5,0.926319) +- (0,{0.00373645}) (6,0.926917) +- (0,{0.0037479}) (7,0.925417) +- (0,{0.00388754}) (8,0.927215) +- (0,{0.0037878}) (9,0.923918) +- (0,{0.00390597}) (10,0.925742) +- (0,{0.00369374})};%
        \addplot[mark=diamond,error bars/.cd,y dir=both, y explicit] plot coordinates {(1,0.953708) +- (0,{0.00252373}) (2,0.953617) +- (0,{0.00245418}) (3,0.95039) +- (0,{0.00274023}) (4,0.950455) +- (0,{0.00264876}) (5,0.950557) +- (0,{0.00272156}) (6,0.949247) +- (0,{0.00275547}) (7,0.949455) +- (0,{0.00284501}) (8,0.94963) +- (0,{0.0028042}) (9,0.949823) +- (0,{0.00275119}) (10,0.948896) +- (0,{0.00278876})};
      \nextgroupplot[legend entries={{$(0,1)$},{$(0.5,1)$},{$(0.9,1)$}},legend to name=groups_legend1] 
        \addplot[mark=o,error bars/.cd,y dir=both, y explicit] plot coordinates {(1,0.171944) +- (0,{0.00776592}) (2,0.175693) +- (0,{0.00805463}) (3,0.147996) +- (0,{0.00745144}) (4,0.129056) +- (0,{0.00755197}) (5,0.113812) +- (0,{0.00726142}) (6,0.105265) +- (0,{0.0071333}) (7,0.101245) +- (0,{0.00725309}) (8,0.0975165) +- (0,{0.00707782}) (9,0.0916932) +- (0,{0.00723702}) (10,0.0920128) +- (0,{0.00675691}) };%
        \addplot[mark=square,error bars/.cd,y dir=both, y explicit] plot coordinates { (1,0.216644) +- (0,{0.00860441}) (2,0.216679) +- (0,{0.00919461}) (3,0.212992) +- (0,{0.00928847}) (4,0.20942) +- (0,{0.00915695}) (5,0.20926) +- (0,{0.00956893}) (6,0.209179) +- (0,{0.00946109}) (7,0.20805) +- (0,{0.0100509}) (8,0.209957) +- (0,{0.00985011}) (9,0.208697) +- (0,{0.00979965}) (10,0.211365) +- (0,{0.00979453})};%
        \addplot[mark=diamond,error bars/.cd,y dir=both, y explicit] plot coordinates { (1,0.22592) +- (0,{0.00834694}) (2,0.226232) +- (0,{0.008016}) (3,0.226689) +- (0,{0.00858597}) (4,0.22444) +- (0,{0.00834212}) (5,0.224082) +- (0,{0.00823647}) (6,0.224933) +- (0,{0.0085006}) (7,0.224629) +- (0,{0.0083697}) (8,0.222718) +- (0,{0.00849421}) (9,0.224809) +- (0,{0.00822758}) (10,0.223863) +- (0,{0.00845812}) }; %
    \end{groupplot} 
    \draw (my plots0 c2r1.south)+(0pt,-50pt) node {\ref{groups_legend1}};
    \draw (my plots0 c1r1.north)+(0pt,+10pt) node {Survey};
    \draw (my plots0 c2r1.north)+(0pt,+10pt) node {Digg};
    \draw (my plots0 c3r1.north)+(0pt,+10pt) node {Delicious};
  \end{tikzpicture}
  
  \caption{The value reported is the average recall obtained when all peers have the same distribution over privacy weights for all items, plotted against the number of slices.}
  \label{fig:het_priv_steps}
\end{figure}

From Figure~\ref{fig:het_priv_one_group} (Delicious), we can observe that there is not much difference in terms of utility between the situations in which $\underline u=0.5$ and $\underline u=0.9$, as both situations are close to the utility obtained with the baseline algorithm. Indeed, the largest difference is obtained when $\underline u$ is set to $0$, in which case the utility gets closer from the utility obtained through a random clustering. Furthermore, Figure~\ref{fig:het_priv_steps} (Delicious) demonstrates that varying the number of slices has almost no effect on the utility achieved by $\underline u \in \{0.5,0.9\}$, but has significant impact on the situation in which $\underline u=0$, for which the utility decreases as the number of slices increases. One possible interpretation is that as $n$ (the number of slices) increases, there are more and more items whose privacy weight differs from $1$. In a nutshell, this seems to indicate that items with high privacy weights (above $0.5$) have an important impact on the utility. Combining this observation with the fact that when $\underline u\geqslant 0.5$, the utility was not affected show that items with low privacy weights (less than $0.5$) can harm the utility in a non-negligible manner. While this may seem strange at first glance, it actually solves an apparent contradiction when taken from a different point of view. Consider for instance that a privacy weight is changed by $0.1$ (\emph{e.g.}, from $0.9$ to $0.8$ or from $0.1$ to $0.2$). The amount of impact on utility resulting from this change depends not only on the value of the difference, but also on the original value being changed. On one hand, the utility gain resulting from modifying the privacy weight from $0.1$ to $0.2$ is less than the utility loss if the privacy weight were modified from $0.9$ to $0.8$. On the other hand, changing $\underline u$ from $0.5$ to $0$ can cause a significant damage to utility because the average privacy weight drops from $0.75$ to $0.5$.

\subsection{Varying  Privacy Attitudes Among Users}
\label{sec:heterogeneous_groups}

The results of the previous section were obtained for the setting in which all nodes draw their privacy weights from the same distribution (\emph{i.e.}, all users have the same privacy attitude). However, according to a recent survey \cite{DBLP:journals/ijmms/JensenPJ05}, users of information systems can be classified in at least three very different groups called the \emph{Westin categories} \cite{harris2003}. These three groups are: \textsc{Privacy Fundamentalists}, \textsc{Privacy Pragmatists} and \textsc{Privacy Unconcerned}. The first group is composed of the users concerned about their privacy, while on the contrary the third group is composed of the ones that are the least concerned (according to a particular definition of concern detailed in the cited poll) and finally the second group is anything in between. For the following experiments, we have adopted the spirit of this classification and consider the three groups of users defined thereafter.    

Each group is equipped with a different distribution from which they pick their privacy weights as follows.
\begin{enumerate}
\item The \textsc{Unconcerned} group corresponds to users that do not really care about their privacy and thus all their items have a privacy weight of $1$.
\item The \textsc{Pragmatists} group represent users that care a little bit about their privacy, such that all their items have a privacy weight chosen uniformly at random among $\{0.5, 0.75,1\}$.
\item The \textsc{Fundamentalists} group embodies users that really care a lot about their privacy and whose items have a privacy weight chosen uniformly at random among $\{0, 0.5,1\}$.
\end{enumerate}
The main issue we want to investigate is how the presence of a relatively conservative group (\emph{i.e.}, having relatively high privacy attitudes) affect the utility of other groups. More specifically, we want to measure whether or not the presence of a group of nodes with high privacy attitudes indirectly \emph{punish} (\emph{i.e.}, reduce the utility) of other more open groups. 

During the experimentations, we have tried different proportions of these groups for a total number of users of $500$. Each value plotted in Figure~\ref{fig:het_priv_3_groups}, has been averaged over $10$ runs but the partition in groups is fixed for a given set of runs.  All experiments are averaged on $\varepsilon \in \{ 0.1,0.5,1,2,3\}$. According to a 2004 poll \cite{DBLP:journals/ijmms/JensenPJ05}, the percentage of each of the privacy groups \textsc{Fundamentalists}, \textsc{Pragmatists} and \textsc{Unconcerned} are respectively, $34\%$,$43\%$ and $23\%$. Nonetheless, we also experiment a combination of several other distributions in order to investigate other possible settings.
In particular, we have also tried the following percentages for each group: the proportion of the \textsc{Unconcerned} group and \textsc{Pragmatists} group vary in the following range $\{ 10\%, 20\%, 60\%, 70\% \}$, while the \textsc{Fundamentalists} group is assigned to the remaining percentage (\emph{i.e.}, there is only two degrees of freedom). If \textsc{Unconcerned} group + \textsc{Pragmatists} group $> 100\%$, then this combination is discarded. In Figure~\ref{fig:het_priv_3_groups}, the $x$-axis represents the percentage of the \textsc{Fundamentalists} group, while the $y$-axis corresponds to the recall. Each of the three lines correspond to the recall of one of the three groups (\textsc{Fundamentalists}, \textsc{Pragmatists}, and \textsc{Unconcerned}). For each of the four plots, the proportion of the \textsc{Pragmatists} group is denoted in the plot by the expression \emph{Pragmatists = some value}. The proportion of the remaining group (\textsc{Unconcerned}) can be directly inferred by subtracting the proportions of the two other groups from $100\%$.
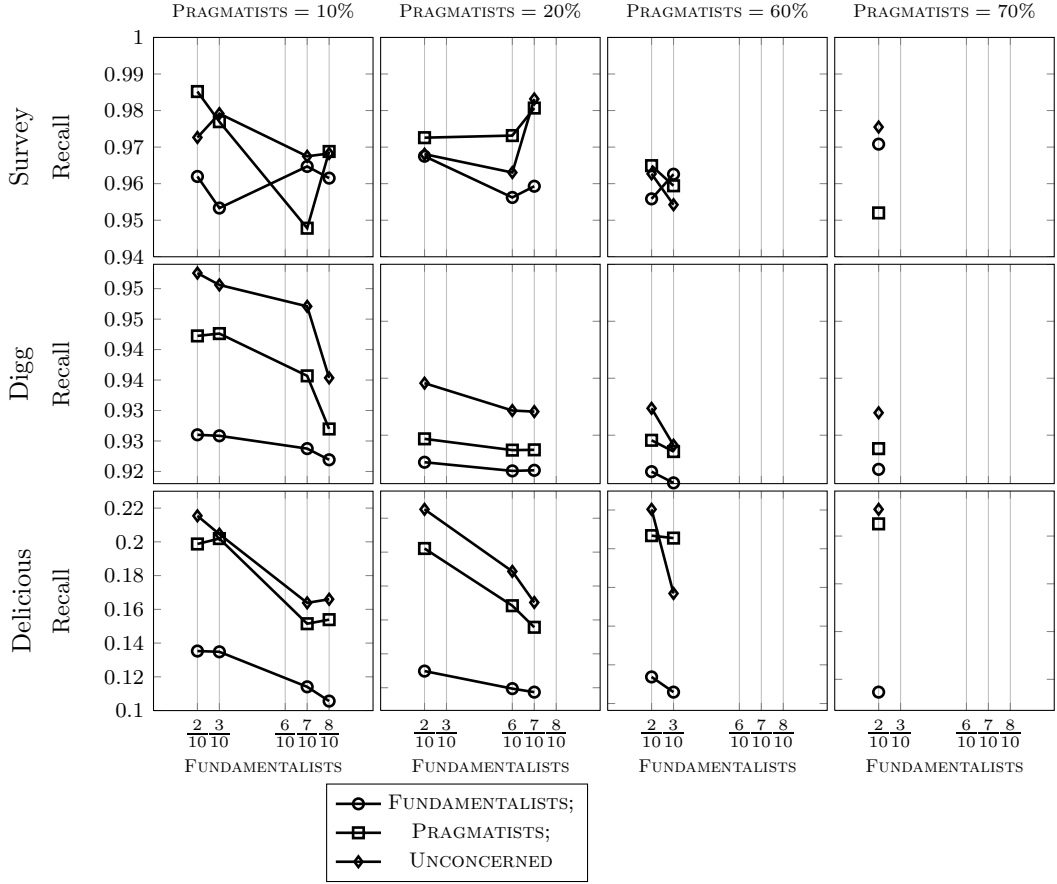
\begin{figure}
  \centering

  \begin{tikzpicture}
    \begin{groupplot}[%
      group style={ group name=my plots, group size=4 by 3, x descriptions at=edge bottom, horizontal
        sep=0.1cm, vertical sep=0.1cm, yticklabels at=edge left, y descriptions at=edge left},%
      footnotesize, width=4.5cm, height=4.5cm, xlabel={\scriptsize\textsc{Fundamentalists}}, ylabel=Recall, xmin=0,xmax=1,
xtick={0.2,0.3,0.6,0.7,0.8},xticklabels={$\frac{2}{10}$,$\frac{3}{10}$,$\frac{6}{10}$,$\frac{7}{10}$,$\frac{8}{10}$}
, xmajorgrids=true]%
      \nextgroupplot[ymax=1,ymin=0.94]
      \addplot[mark=o,line width=1pt] plot coordinates{ (0.8,0.961515) (0.7,0.964721) (0.3,0.953323) (0.2,0.961969) }; %
      \addplot[mark=square,line width=1pt] plot coordinates{ (0.8,0.968804) (0.7,0.947818) (0.3,0.976957) (0.2,0.985202) }; %
      \addplot[mark=diamond,line width=1pt] plot coordinates{ (0.8,0.968259) (0.7,0.967482) (0.3,0.979159) (0.2,0.972639) }; %
      \nextgroupplot[ymax=1,ymin=0.94]
      \addplot[mark=o,line width=1pt] plot coordinates{ (0.7,0.959294) (0.6,0.956225) (0.2,0.9675) }; %
      \addplot[mark=square,line width=1pt] plot coordinates{ (0.7,0.980684) (0.6,0.973155) (0.2,0.972579) }; %
      \addplot[mark=diamond,line width=1pt] plot coordinates{ (0.7,0.983174) (0.6,0.963049) (0.2,0.968067) }; %
      \nextgroupplot[ymax=1,ymin=0.94]
      \addplot[mark=o,line width=1pt] plot coordinates{ (0.3,0.962581) (0.2,0.955841) }; %
      \addplot[mark=square,line width=1pt] plot coordinates{ (0.3,0.959421) (0.2,0.964924) }; %
      \addplot[mark=diamond,line width=1pt] plot coordinates{ (0.3,0.954234) (0.2,0.962693) }; %
      \nextgroupplot[ymax=1,ymin=0.94]
      \addplot[mark=o,line width=1pt] plot coordinates{ (0.2,0.970811) }; %
      \addplot[mark=square,line width=1pt] plot coordinates{ (0.2,0.952) }; %
      \addplot[mark=diamond,line width=1pt] plot coordinates{ (0.2,0.975486) }; %
      \nextgroupplot[ymax=0.959,ymin=0.923]
      \addplot[mark=o,line width=1pt] plot coordinates{ (0.8,0.926913) (0.7,0.928749) (0.3,0.93085) (0.2,0.931017) }; %
      \addplot[mark=square,line width=1pt] plot coordinates{ (0.8,0.931971) (0.7,0.940688) (0.3,0.947629) (0.2,0.947245) }; %
      \addplot[mark=diamond,line width=1pt] plot coordinates{ (0.8,0.940358) (0.7,0.952098) (0.3,0.95561) (0.2,0.957555) }; %
      \nextgroupplot[ymax=1,ymin=0.923]
      \addplot[mark=o,line width=1pt] plot coordinates{ (0.7,0.927678) (0.6,0.927472) (0.2,0.930505) }; %
      \addplot[mark=square,line width=1pt] plot coordinates{ (0.7,0.934862) (0.6,0.934764) (0.2,0.938725) }; %
      \addplot[mark=diamond,line width=1pt] plot coordinates{ (0.7,0.948286) (0.6,0.948664) (0.2,0.95827) }; %
      \nextgroupplot[ymax=1,ymin=0.923]
      \addplot[mark=o,line width=1pt] plot coordinates{ (0.3,0.923259) (0.2,0.927202) }; %
      \addplot[mark=square,line width=1pt] plot coordinates{ (0.3,0.934262) (0.2,0.938215) }; %
      \addplot[mark=diamond,line width=1pt] plot coordinates{ (0.3,0.936478) (0.2,0.949434) }; %
      \nextgroupplot[ymax=1,ymin=0.923]
      \addplot[mark=o,line width=1pt] plot coordinates{ (0.2,0.928011) }; %
      \addplot[mark=square,line width=1pt] plot coordinates{ (0.2,0.935287) }; %
      \addplot[mark=diamond,line width=1pt] plot coordinates{ (0.2,0.947866) }; %
      \nextgroupplot[ymax=0.23,ymin=0.1] 
      \addplot[mark=o,line width=1pt] plot coordinates{ (0.8,0.10563) (0.7,0.114083) (0.3,0.134836) (0.2,0.135289) }; %
      \addplot[mark=square,line width=1pt] plot coordinates{ (0.8,0.153896) (0.7,0.151486) (0.3,0.201944) (0.2,0.198741) }; %
      \addplot[mark=diamond,line width=1pt] plot coordinates{ (0.8,0.16598) (0.7,0.163851) (0.3,0.204568) (0.2,0.215418) }; %
      \nextgroupplot 
      \addplot[mark=o,line width=1pt] plot coordinates{ (0.7,0.117407) (0.6,0.119469) (0.2,0.129819) }; %
      \addplot[mark=square,line width=1pt] plot coordinates{ (0.7,0.155669) (0.6,0.168389) (0.2,0.202146) }; %
      \addplot[mark=diamond,line width=1pt] plot coordinates{ (0.7,0.170302) (0.6,0.188638) (0.2,0.225248) }; %
      \nextgroupplot 
      \addplot[mark=o,line width=1pt] plot coordinates{ (0.3,0.125898) (0.2,0.133747) }; %
      \addplot[mark=square,line width=1pt] plot coordinates{ (0.3,0.205493) (0.2,0.206758) }; %
      \addplot[mark=diamond,line width=1pt] plot coordinates{ (0.3,0.17689) (0.2,0.220355) }; %
      \nextgroupplot[legend columns=1, legend entries={\textsc{Fundamentalists};,\textsc{Pragmatists};,\textsc{Unconcerned}}, legend to name=groups_legend3] 
      \addplot[mark=o,line width=1pt] plot coordinates{ (0.2,0.134637) }; %
      \addplot[mark=square,line width=1pt] plot coordinates{ (0.2,0.205108) }; %
      \addplot[mark=diamond,line width=1pt] plot coordinates{ (0.2,0.211198) }; %
    \end{groupplot}

    \draw (my plots c1r1.west)+(-50pt,0pt) node[rotate=90] {Survey};
    \draw (my plots c1r2.west)+(-50pt,0pt) node[rotate=90] {Digg};
    \draw (my plots c1r3.west)+(-50pt,0pt) node[rotate=90] {Delicious};

    \draw (my plots c1r1.north)+(0,10pt) node {\scriptsize\textsc{Pragmatists}$\;=10\%$};
    \draw (my plots c2r1.north)+(0,10pt) node {\scriptsize\textsc{Pragmatists}$\;=20\%$};
    \draw (my plots c3r1.north)+(0,10pt) node {\scriptsize\textsc{Pragmatists}$\;=60\%$};
    \draw (my plots c4r1.north)+(0,10pt) node {\scriptsize\textsc{Pragmatists}$\;=70\%$};

  \end{tikzpicture}
\ref{groups_legend3}
  \caption{Results obtained for the Delicious, Digg and survey datasets. The heterogeneous differential privacy has been computed for $3$ groups with different privacy attitudes. For a particular figure and a particular $x$ tick, the percentage of \textsc{Unconcerned} group is fully determined as $(1-$ \textsc{Pragmatists} $-$ $ x $).}
  \label{fig:het_priv_3_groups}
\end{figure}

From the results obtained, we can conclude that (1) \textsc{Pragmatists} and \textsc{Unconcerned} always have better recall than \textsc{Fundamentalists} and (2) \textsc{Unconcerned} often have a better recall than \textsc{Pragmatists}, though not always. This seems to indicate that the group caring more about privacy usually is punished more (\emph{i.e.}, its utility is low) than groups that are more liberal with respect to privacy expectations. This not really surprising as a low privacy weight will result in users from the \textsc{Fundamentalists} group segregating themselves from other users in the clustering to the point that they will not necessarily have meaningful neighbors in their view. Finally, to the question whether (or not) more liberal groups will be punished by conservative groups, the answer seems to be negative. Indeed it can be seen from the results of the experiments, that conservative groups are punished more than liberal groups. For instance, the utility of liberal groups only decreases from $0.22$ to $0.19$ as the percentage of conservative groups increases from $20$\% to $80$\%.

\section{Related Work}
\label{sec:related-work}

The majority of previous works on heterogeneous privacy has focused only on user-grained privacy \cite{DBLP:journals/ppna/DasBK11,DBLP:journals/informs/KumarGG10}, in which each user may define his own privacy level (instead of having the same privacy guarantee for all users across the system). As opposed to item-grained privacy, which allows each item of an individual user to have a different privacy weight, user-grained privacy restricts all the items of the same user to the same privacy weight.
For instance, Das, Bhaduri, and Kargupta \cite{DBLP:journals/ppna/DasBK11} have proposed a secure protocol for aggregating sums in a P2P network. In this setting, each node has an input vector, which could be for instance a profile. In a nutshell in this protocol, each node picks at random a few other nodes of the system with whom it computes some local function\footnote{The function is local in the sense that it depends only on the inputs of the node and the peers it has chosen.} in a private manner (the local function begins with a sum as well). The more peers a specific node chooses to participate to the computation, the higher the privacy will be obtained by this node according to the considered definition of privacy. More precisely in their setting, privacy is mainly quantified by the probability of collusion of the peers chosen by a particular node when the aggregation protocol is run. This probability can be made smaller by choosing a larger set of peers, the main intuition being that for a particular node running the aggregation protocol with a larger group diminishes the probability that all these nodes will collude against him. Thus, the best privacy guarantees could be obtained by running the protocol with the entire set of peers but this would be too costly in practice. The main objective is this protocol is to be adaptive by providing a trade-off between the privacy level chosen by a user and the resulting cost in terms of computation and communication. In particular, each user has the possibility to choose heterogeneously the peers with whom he wants to run the aggregation protocol by taking into account his own privacy preferences. However, this work does not seem to be easily extendable to integrate item-grained privacy.

Another work due to Kumar, Gopal, and Garfinkel \cite{DBLP:journals/informs/KumarGG10} is a form of generalization of $k$-anonymity \cite{DBLP:journals/ijufks/Sweene02}. The standard definition of $k$-anonymity requires that in the sanitized database that is released, the profile of a particular individual should be indistinguishable from at least $k-1$ other individuals (thus here $k$ can be considered as being the privacy parameter). The proposed generalization \cite{DBLP:journals/informs/KumarGG10} essentially enables each user to require a different value for $k$ for each attribute in his profile. For example, a user may require that his data should be included in the published database only if there are at least $4$ other users sharing his ZIP code and at least $8$ other users whose age difference with him is at most $3$ years. The possibility of setting the range of a particular attribute could be regarded as item-grained heterogeneous privacy in the sense that an attribute whose privacy range is large is less likely to be useful for de-anonymizing the user than less private attribute. To summarize, the main objective of this approach is to protect the privacy of a user by anonymizing it (\emph{e.g.}, to prevent de-anonymization and linking attacks), while in our work the main objective is to prevent the possibility of inferring the presence or absence of a particular item in the profile.

A line of research on auctions \emph{for} privacy has provided almost the same definition for the heterogeneous differential privacy as ours
\cite{DBLP:conf/sigecom/GhoshR11,DBLP:journals/corr/abs-1111-2885}.
The main difference with our contribution is that these previous works do not provide a mechanism to realize heterogeneous difference privacy, but instead only use the definition to achieve the post-release privacy guarantees. 
In the model studied, the participants are composed of a data analyst and a group of users. Each user has as input a private bit and the data analyst wants to estimate in a differentially-private manner a global function of the private bits of all users, such as the sum or the weighted sum. The data analyst is willing to pay each user for the loss of privacy he incurred by participating to this process. More precisely, each user $i$ has a \emph{privacy valuation} $v_i(\varepsilon_i) : \mathbb R_+ \to \mathbb R_+$ indicating the amount of his loss given the privacy guarantee he gets. The user has no control over $\varepsilon_i$ (\emph{i.e.}, the privacy guarantee he ends up with), which is decided solely by the auction mechanism. As such the valuation function $v_i$ merely affects the payment of the user, as his payment is decided indirectly by the mechanism given the valuation function and is not decided directly by him. Therefore, our work is incomparable to theirs, because the privacy parameter $\varepsilon_i$ acts mainly as an indication about the level of privacy reached, while in our setting the privacy parameter represents the \emph{user's requirement} about the privacy of a particular item of his profile. Moreover, in \cite{DBLP:conf/sigecom/GhoshR11} it is stated that users finally end up having completely homogeneous privacy guarantees. More precisely, each user ends up having $\varepsilon$-differential privacy, with some $\varepsilon$ being the same for all users. In contrast in \cite{DBLP:journals/corr/abs-1111-2885}, users effectively have heterogeneous privacy guarantees. However, these guarantees are determined by the public weights of the auctioneer, which the auctioneer chooses so as to compute the weighted average of the users' inputs and independently of the privacy valuations of the users.

Finally, Nissim, Raskhodnikova and Smith \cite{DBLP:conf/stoc/NissimRS07} have investigated how the amount of noise necessary to achieve differential privacy can be tailored by taking into account to the particular inputs (\emph{i.e.}, profiles) of participants, in addition to the sensitivity of the function considered. The main objective of this approach is to reduce the amount of noise that needs to be added to inputs that are not \emph{locally sensitive} (\emph{i.e.}, for which the output does not change much if only one item is changed). However, they also show that the amount of noise added may itself reveal information about the inputs. Hence, they defined a differentially private version formalizing the notion of local sensitivity called \emph{smooth sensitivity}, guaranteeing that the amount of noise added is itself $\varepsilon$-differentially private. Similarly, we have ensured that for our notion of heterogeneous differential privacy, the amount of noise added is not impacted by the specific profile considered or by the privacy requirements formulated by a user. Rather, we have modified the function under consideration and its sensitivity, which also impacts the distortion induced of the output (\emph{cf.}, Section~\ref{sec:distortion}). We have also proven that the privacy requirements of a user expressed in the form of \emph{private weights} remain private are they are also covered by $\varepsilon$-differentially privacy guarantees. Thus, it is difficult for an adversary observing the output of an heterogeneous differentially private mechanism to guess the privacy weight that a user has put on a particular item of his profile.

\section{Conclusion}
\label{sec:conclusion}

In this work, we have introduced the novel concept of \emph{heterogeneous differential privacy} that can accommodate for different privacy expectations not only per user but also per item as opposed to previous models that implicitly assume uniform privacy requirements. We have also described a generic mechanism achieving HDP called the \emph{Stretching Mechanism}, which protects at the same time the items of the profile of user and the privacy vector representing his privacy expectations across items of the profile. We applied this mechanism for the computation of the cosine similarity and evaluate its impact on a distributed semantic clustering task by using the recall as a measure of utility. Moreover, we have conducted an experimental evaluation of the impact of having different groups of users with different privacy requirements.

Although the Stretching Mechanism can be applied to a wealth of functions, it is nonetheless not directly applicable to some natural functions, such as the $\ell_0$ norm and $\min$. Indeed, when computing the $\ell_0$ norm (\emph{i.e.}, the number of non-zero coordinates in a given vector), each coordinate contributes either zero or one regardless of its value. Since the Stretching Mechanism modifies this value, this mechanism would always output the true exact value as long as no privacy weight has been set to exactly zero. For the case of $\min$, due to the fact that the Stretching Mechanism shrinks each coordinate by a factor corresponding to its privacy weight, the resulting output may not have anymore a relation to the intended semantics of the function $\min$. 

Another challenge is to enable users to estimate the amount of distortion in the output that they received out of an heterogenous differentially private mechanism. For instance, for functions such as the sum, recipients will not be able to estimate the correct value without being given the distortion. Although the distortion has an upper bound given by Theorem~\ref{thm:distortion}, the information needed to compute the upper bound is private. Therefore, releasing the distortion (or even its upper bound) would constitute a violation of privacy. We believe this issue could be solved partially by releasing an upper bound using the traditional Laplacian mechanism at an additional cost of an $\varepsilon$ amount of privacy.
Another important future work includes the characterization of functions that have a low and high distortion. Indeed, functions having a high distortion are not really suitable for our HDP mechanism. We also leave as open the question of designing a different mechanism than the Stretching Mechanism achieving HDP with a lower distortion.

\bibliographystyle{IEEEtran}
\bibliography{refs.bib}
 
\appendix

\section{Proofs}
 
\subsection{Proof of Lemma~\ref{lem:gsi_leq_psii_gs}}

\begin{proposition}[Monotonicity of subdomain optimization]
\label{proposition:maximization_on_smaller_domain}
Let $\theta$ and $\theta^{\prime}$ be the result of two maximization problems $p_1$ and $p_2$ of the function $g$ in which the maximization is over domains $J$ and $J^{\prime}$, respectively. Then, if $J\subseteq J^{\prime}$, this implies that $\theta \leqslant \theta^{\prime}$. The opposite statement also holds for minimization problems.
\end{proposition}
\begin{proof}
  Since $\theta^{\prime}$ is the optimal result of $p_2$ over $J^{\prime}$, this means that by definition:
  \begin{equation}
    \label{eq:max_j}
    g(\theta^{\prime}) \geqslant g(j) \mbox{ , for all } j \mbox{ in }J^{\prime}.
  \end{equation}  
Moreover, since any result $\theta$ for $p_1$ will always be in $J$, and therefore in $J^{\prime}$, then $g(\theta^{\prime}) \geqslant g(\theta)$ by \eqref{eq:max_j}. (The proof that the opposite statement holds for minimization problems follows from the same arguments and thus we choose to omit it.) %
\end{proof}

\begin{lemma}[Shrinkage matrices composition]
  \label{lem:two_shrinking_matrices}
  If $A$ and $B$ are two shrinkage matrices and $D$ a semi-balanced set, then $ABD \subseteq BD \subseteq D$.
\end{lemma}
\begin{proof}
By definition of semi-balanced set, we have $BD \subseteq D$. Then it remains to prove that $ABD\subseteq BD$ (or equivalently, that $BD$ is a semi-balanced set). We observe that a vector $\vec w$ belongs to $ABD$ if and only if $\vec w=AB\vec b$ for some $\vec b\in D$. Because shrinking matrices commute, $\vec w=BA\vec b$. Let $\vec a = A\vec b$. By definition of semi-balanced set, $\vec a \in D$. Therefore, $\vec w = B \vec a$ for $\vec a \in D$, which means $\vec w$ belongs to $BD$ by definition of $BD$.
\end{proof}
 
\begin{lemma}[Monotonicity of the global sensitivity]
\label{lemma:w1_leq_w2_gw1_leq_gw2}
If $\vec w^\prime \leqslant \vec w$ then $S(R,\vec w^\prime)\leqslant
S(R,\vec w)$.
\end{lemma}
\begin{proof}
  Let $\vec c$ be such that
  \begin{equation}
  c_i =
  \begin{cases}
    w_i^\prime / w_i & \mbox{if } w_i \neq 0 \\
    0       & \mbox{otherwise}
  \end{cases},
\end{equation}
  and let $C = \diag(\vec c)$ is a shrinkage matrix. Then $\vec w^\prime = C \vec w$. Let $T^\prime=\diag(\vec w^\prime)$ and $T = \diag(\vec w)$ be two other shrinkage matrices. Notice that $T^\prime=CT$. By Lemma~\ref{lem:two_shrinking_matrices} and since $D$ is semi-balanced:
  \begin{equation}
  T^\prime D = C T D \subseteq T D \subseteq D.
\end{equation}
  
  The result follows from Proposition~\ref{proposition:maximization_on_smaller_domain} because $S(R,\vec w)$ is over the domain $T D$ while $S(R,\vec w^\prime)$ is a maximization problem over the domain $T^\prime D \subseteq T D$. %
\end{proof}

\begin{corollary}[Monotonicity of the modular global sensitivity]
\label{cor:w1_leq_w2_giw1_leq_giw2}
If $\vec w^\prime \leqslant \vec w$ then $S_i(R,\vec
w^\prime)\leqslant S_i(R,\vec w)$ for all $i$.
\end{corollary}
\begin{proof}
  By Lemma~\ref{lemma:w1_leq_w2_gw1_leq_gw2}, we have that $S(R,\vec w^\prime)$ $\leqslant$ $ S(R,\vec w)$.  Let $i^*$ $=$ $\argmax_i S_i(R
,\vec w)$ and therefore $S(R,\vec w)$ $=$ $S_{i^*}(R
,\vec w)$.  In order to get a contradiction, we assume that $S_{i^*}(R,\vec w^\prime)$ $> $ $ S_{i^*}(R,\vec w)$, thus we have 
\begin{equation}
S(R,\vec w^\prime) = \max_i S_i(R,\vec w^\prime) > S_{i^*}(R
,\vec w) = S(R,\vec w) ,
\end{equation}
which is a contradiction.%
\end{proof}

\begin{proof}[Proof of Lemma~\ref{lem:gsi_leq_psii_gs}] 
  Since $\vec w \leqslant (\vec 1_{-i},w_i)$ for all $i$, then: 
  \begin{align}
    S_i(R,\vec w)  & \leqslant S_i(R, (\vec 1_{-i},w_i)) \mbox{ for all } i,\\
    & \leqslant v_i S(f) \quad \text{for all $i$},
  \end{align}
where the first inequality follows by Corollary~\ref{cor:w1_leq_w2_giw1_leq_giw2} and the second inequality follows from the premise of the lemma, thus concluding the proof. %
\end{proof}  

\subsection{Proof of Theorem~\ref{thm:priv-priv-sett}}

\begin{proposition}[Semi-balanced sets are closed under shrinkage.] \label{cor:AD_in_D}
If $D$ is a semi-balanced set and $A$ is a shrinkage matrix, then $AD$ is also a semi-balanced set.
\end{proposition}
\begin{proof}
$AD$ is semi-balanced if and only if for any shrinkage matrix $B$ the following is true: $B(AD)\subseteq AD$. Indeed, since $B$ is a shrinkage matrix and $D$ is a semi-balanced set, then $BD\subseteq D$. By multiplying both sides by a shrinkage matrix $A$, we obtain $ABD \subseteq AD$, which by the fact that $A$ and $B$ are commutative implies that $BAD\subseteq AD$. %
\end{proof}

\begin{proposition}[$T(\vec v)$ and $T(\vec v^{(i)})$ are neighbors]
\label{prop:tvd_tvid_neigh}
$T(\vec v) . \vec d$ and $T(\vec v^{(i)}) \cdot \vec d$ are neighbors on item $i$, since $w_i$ is a function of $v_i$ only.
\end{proposition}

\begin{proof}[Proof of Theorem~\ref{thm:priv-priv-sett}]
 Let $\vec d_* = T(\vec v) \cdot \vec d$ and $\vec d_*^{(i)} = T(\vec v^{(i)}) \cdot \vec d$. By Proposition~\ref{prop:tvd_tvid_neigh} (\emph{cf.} Section~\ref{sec:constructing-tv}), $\vec d_*$ and $\vec d_*^{(i)}$ are neighboring profiles. Moreover due to Proposition~\ref{cor:AD_in_D}, they still belong to $D$, thus $\lvert f(\vec d_*) -  f(\vec d_*^{(i)})\rvert \leqslant S_i(f) \leqslant S(f)$. Therefore, we have:
  \begin{align}
    \frac{\Pr[ \hat{f} (\vec d,\vec v,\varepsilon) = t ]}{\Pr[ \hat{f}
      (\vec d,\vec v^{(i)},\varepsilon) = t ]} & = \frac{h(t-R(f,\vec
      v)(\vec d))}{h(t-R(f,\vec v^{(i)})(\vec d))} \\ & \leqslant \exp
    ( \frac{\varepsilon \lvert R(f,\vec v)(\vec d) - R(f,\vec
      v^{(i)})(\vec d)\rvert}{S(f)} ) \\ & = \exp(\frac{\varepsilon
      \lvert f(\vec d_*) - f(\vec d_*^{(i)}) \rvert}{S(f)}) \\ & \leqslant
    \exp (\frac{\varepsilon \cancel{S(f)}}{\cancel{S(f)}}) = \exp
    (\varepsilon),
  \end{align}
where $h(\cdot)$ is defined in \eqref{eq:h}, thus proving the result.%
\end{proof}

\subsection{Proof of Theorem~\ref{thm:hdp}}

\begin{proof}[Proof of Theorem~\ref{thm:hdp}]
For all two neighboring profiles $\vec d,\vec d^{(i)}$, and for all outputs $t \in \mathbb R$ of the function $f$ we have
\begin{align} 
    \frac{\Pr[ \hat{f} (\vec d,\vec v,\varepsilon) = t ]}{\Pr[ \hat{f} (\vec d^{(i)},\vec v,\varepsilon) = t ]} & = \frac{h(t-R(f,\vec v)(\vec d))}{h(t-R(f,\vec v)(\vec
      d^{(i)}))} \\ & \leqslant \exp ( \frac{\varepsilon \lvert R(f,\vec v)(\vec d) -
      R(f,\vec v)(\vec d^{(i)})\rvert}{S(f)}) \\
  &  \leqslant \exp ( \frac{ \varepsilon S_i(R(f,\vec v))}{S(f)}) \\ & \leqslant \exp ( \frac{ \varepsilon v_i \cancel{S(f)}}{\cancel{S(f)}}) = \exp( \varepsilon v_i),
  \end{align}
where $h(\cdot)$ is defined in \eqref{eq:h},
  thus proving the result.%
\end{proof} 

\subsection{Proof of Theorem~\ref{thm:distortion}}
\begin{proof}[Proof of Theorem~\ref{thm:distortion}]
Let $\vec y=\vec d$ and $\vec x=T(\vec v)\cdot\vec d$, then by the
mean value theorem \cite[Theorem~14.4, p.~301]{linAlgInAct}, there exists a constant $0\leqslant c \leqslant 1$ (depending on $\vec d,T(\vec v)$, and $f$) such that:
\begin{equation}
f(\vec y) - f(\vec x) = \nabla f((1-c) \vec x + c \vec y) \cdot (\vec y-\vec x),
\end{equation}
where $\cdot$ denotes the scalar product. Therefore, by the Cauchy-Schwarz inequality, we have:
\begin{equation}
 \lvert f(\vec y) - f(\vec x) \lvert \leqslant \lVert \nabla f((1-c) \vec x + c \vec y) \rVert \lVert \vec y-\vec x \rVert.
\end{equation}

Finally, by the fact that
\begin{align}
\lVert \vec y - \vec x\rVert & = \lVert \vec d - T(\vec v)\cdot\vec d \rVert = \lVert (I-T(\vec v))\cdot\vec d \rVert \\ & = \sqrt{ \sum_i ((1-w_i) d_i )^2 } \\
& \leqslant (1-\min_i w_i) \sqrt{ \sum_i d_i^2 } = (1-\underline w) \lVert
\vec d \rVert, 
\end{align}
where $\underline w$ is the minimum of $\vec w$ (the diagonal of $T(\vec v)$),and that
\begin{align}
(1-c)\vec x + c \vec y & = (1-c) T(\vec v)\cdot \vec d + c \vec d
    \\ &
 = (c I + (1-c) T(\vec v)) \cdot \vec d,
\end{align}
the theorem follows directly.%
\end{proof}

\subsection{Proof of Lemma~\ref{lem:diag_v_premis_thm_hdp}}

\begin{proof}[Proof of Lemma~\ref{lem:diag_v_premis_thm_hdp}]
 Each profile being represented as a binary vector, the global sensitivity of the scalar product is one (\emph{i.e.}, $S(\mathsf{SP})=1$). Thereafter, for the sake of simplicity, let $R$ denotes $R(\mathsf{SP},\vec v)$. As $T(\vec v)$ is a diagonal matrix, it is strictly identical to its transpose $T(\vec v)^\top$. We can assume without loss of generality that $\vec d^{(i)}=(\vec x,\vec y^{(j)})$ for item $j=i-\dim(x)$, and therefore that:
\begin{align}
S_i(R) & = \max\limits_{\vec d \sim\vec d^{(i)}} \lvert T(\vx) \vec x \cdot
T(\vy) \vec y -   T(\vx) \vec x \cdot
T(\vy) \vec y^{(j)} \rvert \\ 
&=\max\limits_{\vec d \sim \vec d^{(i)}} \lvert  (\vec x^\top T(\vx)  
T(\vy)) \cdot (\vec y - \vec y^{(j)}) \rvert,
\end{align}
however the vector $\vec y - \vec y^{(j)}$ has all its coordinates set to $0$ except for the $j^{th}$ coordinate. Therefore, the maximum is reached when $y_j=1$, $\vec x=\vec 1=(1,\cdots,1)$, and is such that:
\begin{equation}
\vx_j \vy_j \leqslant v_i =  v_i \times 1 = v_i S(\mathsf{SP}),
\end{equation}
which concludes the proof.%
\end{proof}

\subsection{Proof of Lemma~\ref{lem:post-processing}}

\begin{proof}[Proof of Lemma~\ref{lem:post-processing}]
The theorem is equivalent to prove that for any two neighboring profiles $\vec d \sim \vec d^{(i)}$ the following holds:
\begin{equation}
 \Pr[g\circ\hat f(\vec d) = t] \leqslant \exp(\varepsilon v_i) \Pr[g\circ\hat f(\vec d^{(i)}) = t].
\end{equation}
  To prove this, consider any two neighboring profiles $\vec d \sim \vec d^{(i)}$:
  \begin{align}
&  \Pr[g\circ\hat f(\vec d) = t]  = \int\limits_{s\in
      \mathsf{Range}(\hat f)}
    \Pr[\hat f(\vec d) = s]\cdot\Pr[g(s)=t] \\
    & \leqslant \int\limits_{s\in \mathsf{Range}(\hat f)}
    \exp(\varepsilon v_i)
    \Pr[\hat f(\vec d^{(i)}) = s]\cdot\Pr[g(s)=t] \\
    & = \exp(\varepsilon v_i) \Pr[g\circ\hat f(\vec d^{(i)}) = t],
  \end{align}
thus concluding the proof.%
\end{proof}

\end{document}